\documentclass[12pt]{SelfArx} 
\usepackage[english]{babel} 
\usepackage{listings}
\usepackage{lscape}
\usepackage{amsfonts,amssymb,amsbsy,amsmath,amsthm}

\theoremstyle{definition}

\newtheorem{algorithm}{Algorithm}
\newtheorem{definition}{Definition}
\theoremstyle{plain}
\newtheorem{theorem}{Theorem}
\newtheorem{lemma}{Lemma}

\definecolor{color1}{RGB}{0,0,90} 
\definecolor{color2}{RGB}{0,20,20} 

\usepackage{hyperref} 
\hypersetup{hidelinks,colorlinks,breaklinks=true,urlcolor=red,citecolor=color1,linkcolor=color1,bookmarksopen=false,
	pdftitle={Prediction intervals for random-effects meta-analysis: a confidence distribution approach},
	pdfauthor={Kengo Nagashima, Hisashi Noma, Toshi A. Furukawa}}

\JournalInfo{\emph{Statistical Methods in Medical Research} 2019;\textbf{28}(6):1689--1702.} 
\Archive{DOI: 10.1177/0962280218773520} 

\PaperTitle{Prediction intervals for random-effects meta-analysis: a confidence distribution approach}
\ShortTitle{Prediction intervals for random-effects meta-analysis}

\Authors{Kengo Nagashima\textsuperscript{1}*, Hisashi Noma\textsuperscript{2}, and Toshi A. Furukawa\textsuperscript{3}} 
\affiliation{\footnotesize{\textsuperscript{1}\textit{Research Center for Medical and Health Data Science, The Institute of Statistical Mathematics, Japan.}}}
\affiliation{\footnotesize{\textsuperscript{2}\textit{Department of Data Science, The Institute of Statistical Mathematics, Japan.}}}
\affiliation{\footnotesize{\textsuperscript{3}\textit{Departments of Health Promotion and Human Behavior and of Clinical Epidemiology, Kyoto University Graduate School of Medicine / School of Public Health, Japan.}}}
\affiliation{\footnotesize{*\textbf{Corresponding author}: nshi@keio.jp}}
\affiliation{}
\affiliation{This is a postprint / corrected version (October 2025) of the following article: ``Nagashima K, Noma H, Furukawa TA. Prediction intervals for random-effects meta-analysis: a confidence distribution approach. \emph{Statistical Methods in Medical Research} 2019;\textbf{28}(6):1689--1702. DOI: \href{https://doi.org/10.1177/0962280218773520}{10.1177/0962280218773520}. Copyright \copyright \,[2018] (The Authors). Reprinted by permission of SAGE Publications.}
\affiliation{}
\affiliation{The third paragraph of Section 3 (Page 7) and Figure 1 (Page 8) has been corrected, because in Simulation (i) of Section 3, we generated incorrect random numbers. Therefore, we re-performed a simulation with the correct random numbers and corrected results of Simulation (i). These corrections do not alter the conclusion of the paper. We sincerely apologize for the inconvenience.}

\newcommand{\keywordname}{Keywords} 

\Abstract{
\footnotesize{
Prediction intervals are commonly used in meta-analysis with random-effects models.
One widely used method, the Higgins--Thompson--Spiegelhalter (HTS) prediction interval,  replaces the heterogeneity parameter with its point estimate, but its validity strongly depends on a large sample approximation.
This is a weakness in meta-analyses with few studies.
We propose an alternative based on bootstrap and show by simulations that its coverage is close to the nominal level, unlike the HTS method and its extensions.
The proposed method was applied in three meta-analyses.
}
}

\begin{document}

\flushbottom 

\maketitle 


\thispagestyle{empty} 


\section{Introduction}
\label{sec:1}
Meta-analysis is an important tool for combining the results of a set of related studies.
A common objective of meta-analysis is to estimate an overall mean effect and its confidence interval \cite{Borenstein2009}.
Fixed-effect models and random-effects models have been widely applied.

Fixed-effect models assume that treatment effects are equal for all studies.
The estimate of the common treatment effect and its confidence interval provide valuable information for applying the results to a future study or a study not included in the meta-analysis.
By contrast, random-effects models assume that the true treatment effects differ for each study.
The average treatment effect across all studies and its confidence interval have been used together with heterogeneity measures that are important for generalizability.
For instance, the $I^2$-statistic \cite{Higgins2002,Higgins2003} has been widely used as a heterogeneity measure.
However, researchers often interpret results of fixed-effect and random-effects models in the same manner\cite{Riley2011a,Riley2011b}.
They tend to focus on the average treatment effect estimate and its confidence interval.
It is necessary to consider that the treatment effects in each study may be different from the average treatment effect.
Higgins \emph{et al.} \cite{Higgins2009} proposed a prediction interval for a treatment effect in a future study.
This interval can be interpreted as the range of the predicted true treatment effect in a new study, given the realized (past) studies.
A prediction interval naturally accounts for the heterogeneity, and helps apply the results to a future study or a study not included in the meta-analysis.
Riley \emph{et al.} \cite{Riley2011a} recommended that a prediction interval should be reported alongside a confidence interval and heterogeneity measure.

Poor coverage of the confidence intervals in random-effects meta-analysis has been studied extensively\cite{Higgins2009,Brockwell2001,Noma2011}, especially in the context of synthesis of few studies\cite{Kontopantelis2013} (fewer than 20).
Recently, Partlett and Riley \cite{Partlett2017} confirmed that prediction intervals based on established methods, including the Higgins--Thompson--Spiegelhalter (HTS) prediction interval \cite{Higgins2009}, also could have poor coverage.
No explicit solution to this problem has been found thus far.

The HTS prediction interval has a fundamental problem.
It can be regarded as a plug-in estimator that replaces the heterogeneity parameter $\tau^2$ with its point estimate $\hat{\tau}^2$.
The $t$ distribution with $K-2$ degrees of freedom is used to approximately account for the uncertainty of $\hat{\tau}^2$, where $K$ is the number of studies.
Replacement with the $t$-approximation has a detrimental impact on the coverage probability, especially when $K$ is small, as is often the case in practice.
We also confirmed in Section \ref{sec:3} the HTS prediction intervals suffer from severe under-coverage.

In this article, we develop a new prediction interval that is valid under more general and realistic settings of meta-analyses in medical research, including those whose $K$ is especially small.
To avoid using a plug-in estimator, we propose a parametric bootstrap approach using a confidence distribution to account for the uncertainty of $\hat{\tau}^2$ with an exact distribution estimator of $\tau^2$ \cite{Schweder2002,Schweder2016,Singh2005,Singh2007,Xie2013}.
A confidence distribution, like a Bayesian posterior, is considered as a distribution function to estimate the parameter of interest in frequentist inference.

This article is organized as follows.
In Section \ref{sec:2}, we review the random-effects meta-analysis and HTS prediction interval, and then present the new method.
In Section \ref{sec:3}, we assess the performance of the HTS prediction interval and proposed prediction interval in simulation studies.
In Section \ref{sec:4}, we apply the developed method to three meta-analysis data sets.
We conclude with a brief discussion.

\section{Method}
\label{sec:2}
\subsection{The random-effects model and the exact distribution of Cochran's $Q$ statistic}
\label{sec:2.1}
We consider the random-effects model
\cite{Higgins2009,Cochran1937,Cochran1954,DerSimonian1986,Whitehead1991}.

Let the random variable $Y_k$ ($k=1, 2, \ldots, K$) be an effect size estimate from the $k$-th study. 
The random-effects model can be defined as
\begin{equation}
\label{equ:1}
\begin{array}{lll}
Y_k &=& \theta_k+\epsilon_k,\\
\theta_k &=& \mu+u_k,
\end{array}
\end{equation}
where $\theta_k$ is the true effect size of the $k$-th study, $\mu$ is the grand mean parameter of the average treatment effect, $\epsilon_k$ is the random error within a study, and $u_k$ is a random variable reflecting study-specific deviation from the average treatment effect.
It is assumed that $\epsilon_k$ and $u_k$ are independent, with $\epsilon_k \sim N(0, \sigma_k^2)$ and $u_k \sim N(0, \tau^2)$, where the within-study variances $\sigma_k^2$ are known and replaced by their efficient estimates \cite{Biggerstaff1997,Biggerstaff2008}, and the across-studies variance $\tau^2$ is an unknown parameter that reflects the treatment effects heterogeneity.

Under the model in \eqref{equ:1}, the marginal distribution of $Y_k$ is a normal distribution with the mean $\mu$ and the variance $\sigma_k^2+\tau^2$.

Random-effects meta-analyses generally estimate $\mu$ to evaluate the average treatment effect and $\tau^2$ to evaluate the treatment effects heterogeneity.
The average treatment effect $\mu$ is estimated by
\[
\frac{\sum_{k=1}^K (\sigma_k^2+\hat{\tau}^2)^{-1} Y_k}{\sum_{k=1}^K (\sigma_k^2+\hat{\tau}^2)^{-1}},
\]
where $\hat{\tau}^2$ is an estimator of the heterogeneity parameter $\tau^2$.
Estimators of $\tau^2$, such as the DerSimonian and Laird estimator \cite{DerSimonian1986}, have been applied \cite{Sidik2007}.
In this paper, we discuss prediction intervals using the DerSimonian and Laird estimator,
\[
\hat{\tau}_{DL}^2=\max[0, \hat{\tau}_{UDL}^2],
\]
and its untruncated version,
\[
\hat{\tau}_{UDL}^2=\frac{Q-(K-1)}{S_1+S_2/S_1},
\]
where $Q=\sum_{k=1}^K v_k(Y_k-\bar{Y})^2$ is Cochran's $Q$ statistic, $v_k=\sigma_k^{-2}$, $\bar{Y}=\sum_{k=1}^K v_k Y_k/\sum_{k=1}^K v_k$, and $S_r=\sum_{k=1}^K v_k^r$ for $r=1, 2$.
Under the model in \eqref{equ:1}, Biggerstaff and Jackson \cite{Biggerstaff2008} derived the exact distribution function of $Q$, $F_Q(q; \tau^2)$, to obtain confidence intervals for $\tau^2$.
Cochran's $Q$ is a quadratic form that can be written $\boldsymbol{\mathrm{Y}}^{\mathrm{T}}\boldsymbol{\mathrm{A}}\boldsymbol{\mathrm{Y}}$, where $\boldsymbol{\mathrm{Y}}=(Y_1, Y_2, \ldots, Y_K)^{\mathrm{T}}$, $\boldsymbol{\mathrm{A}}=\boldsymbol{\mathrm{V}} - \boldsymbol{\mathrm{v}}\boldsymbol{\mathrm{v}}^{\mathrm{T}}/v_+$, $\boldsymbol{\mathrm{V}}=\mathrm{diag}(v_1, v_2, \ldots, v_K)$, $\boldsymbol{\mathrm{v}}=(v_1, v_2, \ldots, v_K)^{\mathrm{T}}$, $v_+ =\sum_{k=1}^K v_k$, and the superscript `T' denotes matrix transposition.
Here and subsequently, $\boldsymbol{\mathrm{Z}}=\boldsymbol{\mathrm{\Sigma}}^{-1/2} (\boldsymbol{\mathrm{Y}}-\boldsymbol{\mathrm{\mu}}) \sim N(\boldsymbol{\mathrm{0}}, \boldsymbol{\mathrm{I}})$, $\boldsymbol{\mathrm{S}}=\boldsymbol{\mathrm{\Sigma}}^{1/2}\boldsymbol{\mathrm{A}}\boldsymbol{\mathrm{\Sigma}}^{1/2}$, $\boldsymbol{\mathrm{\mu}}=(\mu, \mu, \ldots, \mu)^{\mathrm{T}}$, $\boldsymbol{\mathrm{\Sigma}}=\mathrm{diag}(\sigma_1^2+\tau^2, \sigma_2^2+\tau^2, \ldots, \sigma_K^2+\tau^2)$, $\boldsymbol{\mathrm{0}}=(0, 0, \ldots, 0)^{\mathrm{T}}$, and $\boldsymbol{\mathrm{I}}=\mathrm{diag}(1,1,\ldots,1)$.

\begin{lemma}
	\label{lem:1}
	Under the model in \eqref{equ:1}, $Q$ can be expressed as $\boldsymbol{\mathrm{Z}}^{\mathrm{T}}\boldsymbol{\mathrm{S}}\boldsymbol{\mathrm{Z}}$; then $Q$ has the same distribution as the random variable $\sum_{k=1}^K \lambda_k \chi^2_k(1)$, where $\lambda_k \geq 0$ are the eigenvalues of matrix $\boldsymbol{\mathrm{S}}$, and $\chi_1^2(1), \chi_2^2(1), \ldots, \chi_K^2(1)$ are $K$ independent central chi-square random variables each with one degree of freedom.
\end{lemma}

Lemma \ref{lem:1} was proven by Biggerstaff and Jackson \cite{Biggerstaff2008} using the location invariance of $Q$ (e.g., $Q$ can be decomposed as $\sum_{k=1}^K v_k(Y_k-\mu)^2 -v_+ (\bar{Y}-\mu)^2$), and distribution theory of quadratic forms in normal variables\cite{Scheffe1959,Graybill1976,Mathai1992}.

\subsection{The Higgins--Thompson--Spiegelhalter prediction interval}
\label{sec:2.2}
Suppose $\tau^2$ is known, $\hat{\mu} \sim N(\mu, \mathrm{SE}[\hat{\mu}]^2)$ and the observation in a future study $\theta_{new} \sim N(\mu, \tau^2)$, where $\mathrm{SE}[\hat{\mu}]=\sqrt{1/\sum_{k=1}^K w_k}$ is a standard error of $\hat{\mu}$ given $\tau^2$, and $w_k=(\sigma_k^2+\tau^2)^{-1}$.
Assuming independence of $\theta_{new}$ and $\hat{\mu}$ given $\mu$, $\theta_{new}-\mu \sim N(0, \tau^2+\mathrm{SE}[\hat{\mu}]^2)$.
To replace the unknown $\tau^2$ by its estimator $\hat{\tau}_{DL}^2$, the following approximation is used.
If $(K-2)(\hat{\tau}_{DL}^2+\widehat{\mathrm{SE}}[\hat{\mu}]^2)/(\tau^2+\mathrm{SE}[\hat{\mu}]^2)$ is approximately distributed as $\chi^2(K-2)$, then $(\theta_{new}-\hat{\mu})/\sqrt{\hat{\tau}_{DL}^2+\widehat{\mathrm{SE}}[\hat{\mu}]^2} \sim t(K-2)$, where $\widehat{\mathrm{SE}}[\hat{\mu}]=\sqrt{1/\sum_{k=1}^K \hat{w}_k}$ is the standard error estimator of $\hat{\mu}$, and $\hat{w}_k=(\sigma_k^2+\hat{\tau}_{DL}^2)^{-1}$.
By this approximation, the HTS prediction interval is
\[
\left[
\hat{\mu}-t_{K-2}^{\alpha}\sqrt{\hat{\tau}_{DL}^2+\widehat{\mathrm{SE}}[\hat{\mu}]^2},~
\hat{\mu}+t_{K-2}^{\alpha}\sqrt{\hat{\tau}_{DL}^2+\widehat{\mathrm{SE}}[\hat{\mu}]^2}
\right],
\]
where $t_{K-2}^{\alpha}$ is the $100(1-\alpha/2)$ percentile of the $t$ distribution with $K-2$ degrees of freedom.
The $t$-approximation is appropriate only when both the number of studies and heterogeneity variance are large.

Several HTS-type prediction intervals following restricted maximum likelihood (REML) estimation of $\tau^2$ have been proposed by Partlett and Riley\cite{Partlett2017}.
For example, they discussed a HTS-type prediction interval following REML with the Hartung--Knapp variance estimator \cite{Hartung2001} (HTS-HK) that replaces $\hat{\mu}$, $\hat{\tau}_{DL}^2$, and $\widehat{\mathrm{SE}}[\hat{\mu}]^2$ in the HTS prediction interval with $\hat{\mu}_R$, $\hat{\tau}_{R}^2$, and $\widehat{\mathrm{SE}}_{HK}[\hat{\mu}_R]^2$, and a HTS-type prediction interval following REML with the Sidik--Jonkman bias-corrected variance estimator \cite{Sidik2006} (HTS-SJ) that replaces $\hat{\mu}$, $\hat{\tau}_{DL}^2$, and $\widehat{\mathrm{SE}}[\hat{\mu}]^2$ in the HTS prediction interval with $\hat{\mu}_R$, $\hat{\tau}_{R}^2$, and $\widehat{\mathrm{SE}}_{SJ}[\hat{\mu}_R]^2$, where $\hat{\tau}_{R}^2$ is the REML estimator for the heterogeneity variance \cite{Harville1977,Raudenbush1985,Sidik2007} which is an iterative solution of the equation
\[
\hat{\tau}_{R}^2=
\frac{\sum_{k=1}^K \hat{w}_{R,k}^2\{(Y_k - \hat{\mu}_{R})^2 + 1/\sum_{l=1}^K \hat{w}_{R,k} - \sigma_k^2\}}{\sum_{k=1}^K \hat{w}_{R,k}^2}
,
\]
$\hat{w}_{R,k}=(\sigma_k^2+\hat{\tau}_R^2)^{-1}$, $\hat{\mu}_{R}=\sum_{k=1}^K \hat{w}_{R,k} Y_k / \sum_{k=1}^K \hat{w}_{R,k}$, the Hartung--Knapp variance estimator is defined as
\[
\widehat{\mathrm{SE}}_{HK}[\hat{\mu}_{R}]^2=
\frac{1}{K-1}\sum_{k=1}^K \frac{\hat{w}_{R,k}(Y_k-\hat{\mu}_{R})^2}{\sum_{l=1}^K \hat{w}_{R,l}},
\]
the Sidik--Jonkman bias-corrected variance estimator
\[
\widehat{\mathrm{SE}}_{SJ}[\hat{\mu}_{R}]^2=
\frac{\sum_{k=1}^K \hat{w}_{R,k}^2 (1-\hat{h}_k)^{-1} (Y_k-\hat{\mu}_{R})^2}{(\sum_{k=1}^K \hat{w}_{R,k})^2},
\]
and
\[
\hat{h}_k=
\frac{2\hat{w}_{R,k}}{\sum_{l=1}^K\hat{w}_{R,l}}-
\frac{\sum_{l=1}^K \hat{w}_{R,l}^2(\sigma_l^2+\hat{\tau}_{R}^2)}{(\sigma_k^2+\hat{\tau}_{R}^2) \sum_{l=1}^K \hat{w}_{R,l}^2}.
\]
The empirical coverage of the HTS-HK and HTS-SJ prediction intervals is close to the nominal level under large heterogeneity variance and $K \geq 5$\cite{Partlett2017}.

The HTS prediction intervals show severe under-coverage under small heterogeneity variance or for few studies (see Partlett and Riley\cite{Partlett2017} and Section \ref{sec:3}).
We introduce a prediction interval in which uncertainty about $\tau^2$ is accounted for and show that it is valid under a small number of studies.

\subsection{The proposed prediction interval}
\label{sec:2.3}
We address the issue discussed in Section \ref{sec:2.2} by constructing a new prediction interval via a parametric bootstrap with the exact distribution of $\hat{\tau}_{UDL}^2$ by using a confidence distribution (see Section \ref{sec:2.4}).
The proposed method uses an approximation that differs from those used by Higgins \emph{et al.} \cite{Higgins2009}.
The HTS prediction interval essentially combines the following two approximations: $(\hat{\mu}-\mu)/\sqrt{\widehat{\mathrm{SE}}[\hat{\mu}]}$ approximately distributed as $N(0, 1)$, which is often not satisfactory \cite{Hartung1999}, and $(K-2)(\hat{\tau}_{DL}^2+\widehat{\mathrm{SE}}[\hat{\mu}])/(\tau^2+\mathrm{SE}[\hat{\mu}])$ is approximately distributed as $\chi^2(K-2)$.

From now on we make the following assumptions: Let the observation in a future study $\theta_{new} \sim N(\mu, \tau^2)$, $Y_k \sim N(\mu, \sigma_k^2+\tau^2)$ given $\sigma_k^2$ and $\tau^2$, and $\theta_{new}$ and $\bar{\mu}=\sum_{k=1}^{K}w_k Y_k/\sum_{k=1}^{K}w_k$ are independent.
In Hartung \cite{Hartung1999} and Hartung and Knapp \cite{Hartung2001}, it was shown that assuming normality of $Y_k$, $(\mu-\bar{\mu})/\mathrm{SE}_H[\bar{\mu}]$ is $t$-distributed with $K-1$ degrees of freedom, and $\mathrm{SE}_H[\bar{\mu}]$ is stochastically independent of $\bar{\mu}$, where $\mathrm{SE}_H[\bar{\mu}]^2=\frac{1}{K-1}\sum_{k=1}^K \frac{w_k}{w_+}(Y_k-\bar{\mu})^2$, and $w_+=\sum_{k=1}^K w_k$.
By replacing $\tau^2$ in $(\bar{\mu}-\mu)/\mathrm{SE}_H[\bar{\mu}]$ with an appropriate estimate $\hat{\tau}^2$, $(\hat{\mu}-\mu)/\widehat{\mathrm{SE}}_H[\hat{\mu}]$ is approximately $t$-distributed with $K-1$ degrees of freedom, where $\widehat{\mathrm{SE}}_H[\hat{\mu}]^2=\frac{1}{K-1}\sum_{k=1}^K \frac{\hat{w}_k}{\hat{w}_+}(Y_k-\hat{\mu})^2$, and $\hat{w}_+=\sum_{k=1}^K \hat{w}_k$.
This approximation exhibits better performance than $(\hat{\mu}-\mu)/\sqrt{\widehat{\mathrm{SE}}[\hat{\mu}]} \overset{approx.}{\sim} N(0, 1)$, even for a few studies (see Theorem 4.4 of Hartung \cite{Hartung1999}).

The above assumptions and results lead to a system of equations,
\begin{equation}
\label{equ:2}
\left\{
\begin{array}{c}
\displaystyle \frac{\theta_{new}-\mu}{\tau}=Z\\
\displaystyle \frac{\bar{\mu}-\mu}{\mathrm{SE}_H[\bar{\mu}]}=t_{K-1}\\
\end{array}
\right.,
\end{equation}
where $Z \sim N(0, 1)$ and $t_{K-1} \sim t(K-1)$.
Solving for $\theta_{new}$ in \eqref{equ:2} yields
\begin{equation}
\label{equ:3}
\theta_{new}=\bar{\mu}+Z\tau-t_{K-1}\mathrm{SE}_H[\bar{\mu}],
\end{equation}
and the prediction distribution has the same distribution as $\theta_{new}$ (even with $\tau^2$ unknown).
By replacing $\tau^2$ in \eqref{equ:3} with an appropriate estimator (not an estimate), we have
\[
\hat{\theta}_{new}=\hat{\mu}+Z\hat{\tau}_{UDL}-t_{K-1}\widehat{\mathrm{SE}}_H[\hat{\mu}],
\]
and an approximate prediction distribution can be given by the distribution of $\hat{\theta}_{new}$.
We use the untruncated estimator $\hat{\tau}_{UDL}^2$ here, because we do not need the truncation to consider the distribution of an estimator of $\tau^2$.
Hence, $\Pr (c_l<\theta_{new}<c_u) = 1-\alpha$ can be approximately evaluated by the distribution of $\hat{\theta}_{new}$.
Since $\hat{\theta}_{new}$ includes three random components, $\hat{\tau}_{UDL}^2$, $Z$, and $t_{K-1}$, this gives the following algorithm for the proposed prediction interval.\\

\begin{algorithm}
	An algorithm for the proposed prediction interval.
	\begin{enumerate}
		\item Generate $B$ bootstrap samples $\tilde{\tau}_b^2$ ($b=1, \ldots, B$) that are drawn from the exact distribution of $\hat{\tau}_{UDL}^2$, $z_b$ that are drawn from $N(0, 1)$, and $t_b$ that are drawn from $t(K-1)$.
		\item Calculate $\tilde{\mu}_b=\sum_{k=1}^K \tilde{w}_{bk}y_k/\sum_{k=1}^K \tilde{w}_{bk}$, and $\tilde{\theta}_{new,b}=\tilde{\mu}_b+z_b \tilde{\tau}_b-t_b\tilde{\mathrm{SE}}_{H,b}[\tilde{\mu}_b]$, where $\tilde{w}_{bk}=(\sigma_k^2+\tilde{\tau}_b^2)^{-1}$, $\tilde{\mathrm{SE}}_{H,b}[\tilde{\mu}_b]^2=\frac{1}{K-1} \sum_{k=1}^K \frac{\tilde{w}_{bk}}{\tilde{w}_{b+}}(y_k-\tilde{\mu}_b)^2$, and $\tilde{w}_{b+}=\sum_{k=1}^K \tilde{w}_{bk}$.
		\item Calculate the prediction limits $c_l$ and $c_u$ that are $100\times \alpha/2$ and $100\times(1-\alpha/2)$ percentage points of $\tilde{\theta}_{new,b}$, respectively.
	\end{enumerate}
\end{algorithm}

An R package implementing the new method with the three data sets (see Section \ref{sec:4}) and a documentation is available at the publisher's web-site, the CRAN website (https://cran.r-project.org/package=pimeta) and GitHub (https://github.com/nshi-stat/pimeta/).

\subsection{Sampling from the exact distribution of the estimator of $\tau^2$}
\label{sec:2.4}
Confidence distribution is a distribution estimator that can be defined and interpreted in a frequentist framework in which the parameter is a non-random quantity.
A confidence distribution for the parameter of interest $\phi$, as described below, can be easily defined as the cumulative distribution function of a statistic.
The following definition of a confidence distribution was presented in Xie and Singh\cite{Xie2013}.
In the definition, $\Phi$ is the parameter space of the unknown $\phi$, $\boldsymbol{\mathrm{Y}}$ is a random vector, and $\mathcal{Y}$ is the sample space corresponding to sample data $\boldsymbol{\mathrm{y}}=(y_1, y_2, \ldots, y_K)^{\mathrm{T}}$.

\begin{definition}
	\label{def:1}
	(R1) A function $H(\cdot)=H(\boldsymbol{\mathrm{y}},\phi)$ on $\mathcal{Y}\times \Phi \to [0, 1]$ is called a confidence distribution for a parameter $\phi$;
	(R2) If for each given $\boldsymbol{\mathrm{y}} \in \mathcal{Y}$, $H(\cdot)$ is a cumulative distribution function on $\phi$;
	(R3) At the true parameter value $\phi=\phi_0$, $H(\phi_0) \equiv H(\boldsymbol{\mathrm{y}}, \phi_0)$, as a function of the sample $\boldsymbol{\mathrm{y}}$, follows the uniform distribution $U(0, 1)$.
\end{definition}

\noindent
Confidence distribution has a theoretical relationship to the fiducial approach \cite{Fisher1935}, and recent developments \cite{Schweder2002,Schweder2016,Singh2005,Singh2007,Xie2013} have provided useful statistical tools that are more widely applicable than classical frequentist methods.
For example, Efron's bootstrap distribution \cite{Efron1998} is a confidence distribution and a distribution estimator of $\phi$.
In meta-analysis, the $Q$-profile method for an approximate confidence interval for $\tau^2$ \cite{Viechtbauer2007} can be considered as an application of confidence distribution \cite{Schweder2016}.
In this section, we propose the exact distribution of $\hat{\tau}_{UDL}^2$, which is a distribution function for estimating the parameter $\tau^2$ using a confidence distribution, and then develop a method of sampling from the exact distribution.
A useful theorem (Theorem \ref{the:1}) is introduced that provides a condition for confidence distribution.

\begin{theorem}
	\label{the:1}
	If a cumulative distribution function of a statistic, $T(\boldsymbol{\mathrm{Y}})$, is $F_T(T(\boldsymbol{\mathrm{y}}); \phi) \equiv F_T(T(\boldsymbol{\mathrm{Y}})\leq T(\boldsymbol{\mathrm{y}}); \phi)$, and $F_T$ is a strictly monotone (without loss of generality, assume that it is decreasing) function in $\phi$ with the parameter space $\Phi=\{\phi: \phi_{\min}\leq \phi \leq \phi_{\max}\}$ for each sample $\boldsymbol{\mathrm{y}}$, then $H(\phi)=1-F_T(T(\boldsymbol{\mathrm{y}}); \phi)$ is a confidence distribution for $\phi$ that satisfies Definition \ref{def:1}.
\end{theorem}

\begin{lemma}
	\label{lem:2}
	Under the model in \eqref{equ:1}, $H(\tau^2)=1-F_Q(q; \tau^2)$ is a confidence distribution for $\tau^2$.	
\end{lemma}

The proof of Theorem \ref{the:1} is easy and hence is omitted.
Lemma \ref{lem:2} can be easily proved by using Theorem \ref{the:1}, because $F_Q(q; \tau^2)$ is a strictly decreasing function in $\tau^2$ \cite{Jackson2013}.
Note that we use the untruncated version of an estimator of $\tau^2$ with the parameter space $\Phi=[\tau_{\min}^2, \infty]$, and $\tau_{\min}^2$ can be negative.

The proposed algorithm samples from the confidence distribution,  $H(\tau^2)=1-F_Q(q_{obs}; \tau^2)$, where $q_{obs}$ is the observed value of $Q$.
By applying Lemma \ref{lem:2} and the inverse transformation method, if $U$ is distributed as $U(0, 1)$ then $H^{-1}(U)$ follows the distribution $H(\tau^2)$.
A sample $\tilde{\tau}^2=H^{-1}(u)$ can be computed by numerical inversion \cite{Forsythe1977} of $H(\tilde{\tau}^2)=u$, where $u$ is an observed value of the random variable $U$.
If $H(0) > u$, then the sample is truncated to zero ($\tilde{\tau}^2=0$).
It follows from Lemma \ref{lem:1} that $F_Q(q; \tau^2)$ is the distribution function of a positive linear combination of $\chi^2$ random variables.
It can be calculated with the Farebrother's algorithm \cite{Farebrother1984}.

\section{Simulations}
\label{sec:3}
We assessed the properties of the HTS and proposed prediction intervals through simulations.

Simulation data was generated by the random-effects model in \eqref{equ:1}, assuming independent normal errors  $\epsilon_k \sim N(0, \sigma_k^2)$ and $u_k \sim N(0, \tau^2)$.
We conducted three sets of simulations described below.
\begin{enumerate}
	\renewcommand{\labelenumi}{(\theenumi)}
	\renewcommand{\theenumi}{\roman{enumi}}
	\item By reference to Brockwell and Gordon \cite{Brockwell2001,Brockwell2007} and Jackson \cite{Jackson2013}, parameter settings that mimic meta-analyses for estimating an overall mean log odds-ratio were determined.
	The average treatment effect $\mu$ was fixed at $0$, as no generality is lost by setting $\mu$ to zero.
	The across-studies variance was set to $\tau^2=0.01,0.05,0.1,0.2,0.3,0.4$, or $0.5$ \cite{Turner2012,Rhodes2015}; mean $I^2$ values were 29.8\%, 66.0\%, 79.1\%, 88.2\%, 91.8\%, 93.7\%, or 94.9\%, respectively.
	The within-study variances $\sigma_k^2$ were generated from a scaled $\chi^2$ distribution with one degree of freedom, multiplied by $0.25$, and then truncated to lie within $[0.009, 0.6]$.
	The number of studies was set to $K=3,5,10,15,20$, or $25$.
	\item In reference to Partlett and Riley\cite{Partlett2017}, parameter settings were determined to evaluate the empirical performance of prediction intervals under various relative degrees of heterogeneity scenarios.
	The within-study variances $\sigma_k^2$ were generated from $\sigma^2\chi^2(n-1)/(n-1)$, an average within-study variance was set to $\sigma^2=0.1$, and the study sample size was set to $n=30$, where $\chi^2(n-1)$ is a random number from a $\chi^2$ distribution with $n-1$ degrees of freedom.
	The degree of heterogeneity is controlled using the ratio $\tau^2/\sigma^2$.
	The heterogeneity parameter was set to $\tau^2=0.01,0.05,0.1$, or $1$, which corresponds to $\tau^2/\sigma^2=0.1,0.5,1$, or $10$; mean $I^2$ values were 9.1\%, 33.3\%, 50.0\%, or 90.9\%, respectively.
	In addition to the above situation where (ii-a) with all studies of similar size, we consider situations (ii-b) with one large study or (ii-c) with one small study (i.e., a within study variance 10 times smaller or 10 fold higher than the others; one study was randomly selected and the within study variance was set to $\sigma_k^2/10$ or $10\sigma_k^2$).
	The average treatment effect $\mu$ was fixed at $1$.
	The number of studies was set to $K=3,5,7,10,25$, or $100$.
	\item We generated data for $K \times 2 \times 2$ tables using a method similar to that used by Sidik and Jonkman \cite{Sidik2007} and considered meta-analyses based on log odds-ratios.
	The heterogeneity variance was set to $\tau^2=0.01,0.1,0.2,0.4$, or $0.6$.
	The average treatment effect was set to $\mu=0, -0.5$, or $0.5$ to assess the impact of standard errors for odds-ratios.
	The number of studies was set to $K=3,6,12,24,48$, or $96$.
	For each $\tau^2$ and $K$, we first generated $\theta_k$ from $N(\mu, \tau^2)$.
	For each study, the sample sizes were set to be equal $n_{0k}=n_{1k}$, and were randomly selected from integers between 20 and 200.
	The responses of the control group, $X_{0k}$, were generated from a binomial distribution $Bin(n_{0k}, p_{0k})$, and the probability $p_{0k}$ was randomly drawn from a uniform distribution $U(0.05, 0.65)$.
	The responses of the treatment group, $X_{1k}$, were generated from a binomial distribution $Bin(n_{1k}, p_{1k})$ and probability $p_{1k}=p_{0k}\exp\{\theta_k\}(1-p_{0k}+p_{0k}\exp\{\theta_k\})$.
	Finally, we constructed an estimator of $\theta_k$ as $Y_k=\log[X_{1k}(n_{0k}-X_{0k})/\{X_{0k}(n_{1k}-X_{1k})\}]$, its variance estimator as $\hat{\sigma}^2=1/X_{1k} + 1/(n_{1k}-X_{1k})+1/X_{0k}+1/(n_{0k}-X_{0k})$, and we used $\hat{\sigma}^2$ rather than $\sigma^2$.
	If any cells are empty, we added 0.5 to each cell for all $K$ tables.
	Mean $I^2$ values were 7.1\%, 42.3\%, 58.7\%, 73.1\%, or 79.7\%, which correspond to $\tau^2=0.01,0.1,0.2,0.4$, or $0.6$, respectively.
\end{enumerate}
For each setting, we simulated 25\,000 replications.
For each method, two-tailed 95\% prediction intervals were calculated.
The number of bootstrap samples $B$ was set to 5\,000.
The coverage probability was estimated by the proportion of simulated prediction intervals containing the result of a future study $\theta_{new}$ that was generated from a normal distribution $N(\mu, \tau^2)$.

\begin{figure}
	\centering
	\includegraphics[width=0.9\textwidth]{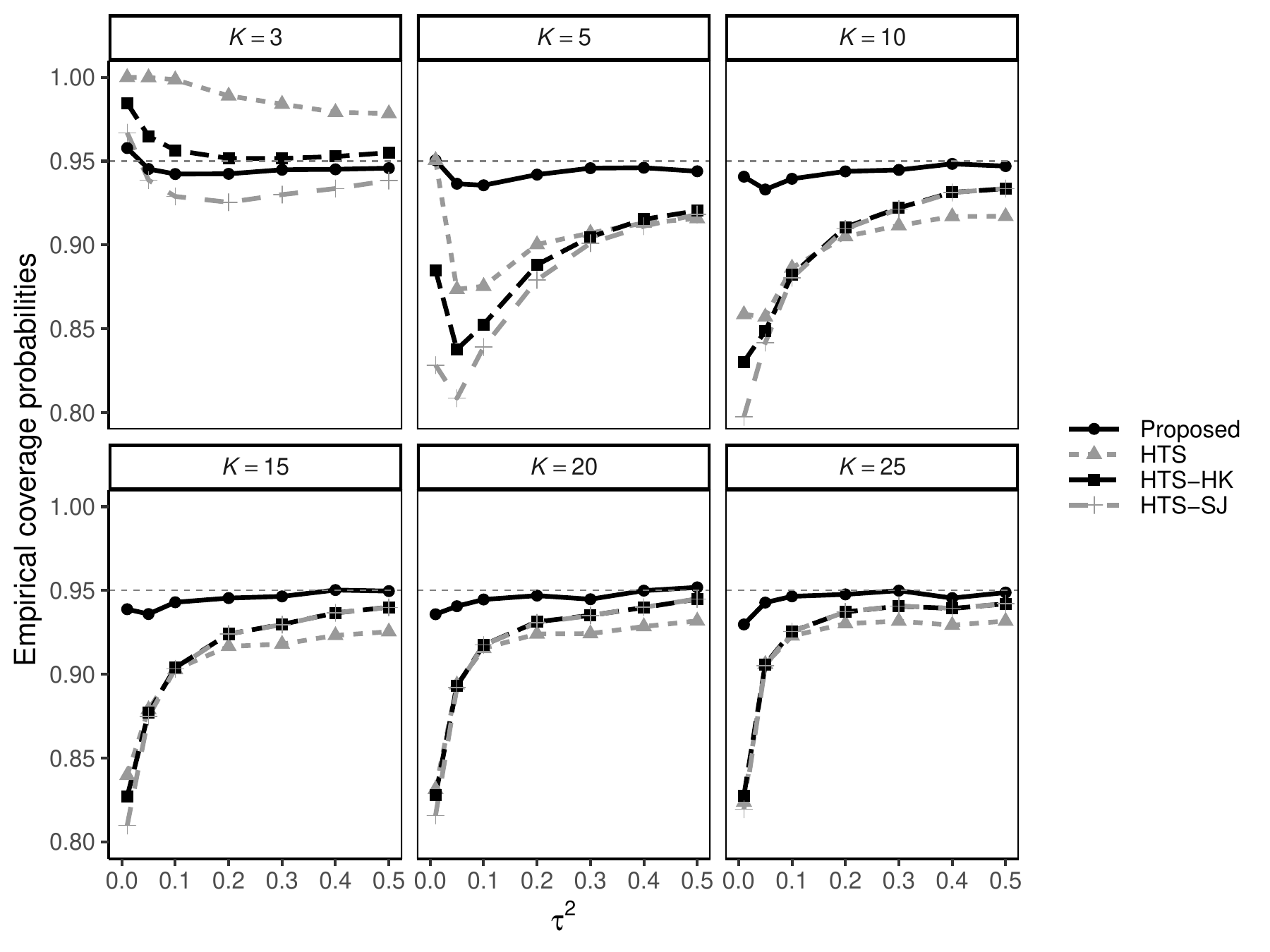}
	\caption{
		Simulation results (i): the performance of the HTS and proposed prediction intervals.
		The number of studies $K=3,5,10,15,20$, or $25$.
		The number of simulations was 25\,000.
		Methods: Proposed, the proposed prediction interval; HTS, the HTS prediction interval; HTS-HK, the HTS-type prediction interval following REML with the Hartung--Knapp variance estimator; HTS-SJ, the HTS-type prediction interval following REML with the Sidik--Jonkman bias-corrected variance estimator.
	}
	\label{fig:1}
\end{figure}

\begin{figure}
	\centering
	\includegraphics[width=0.9\textwidth]{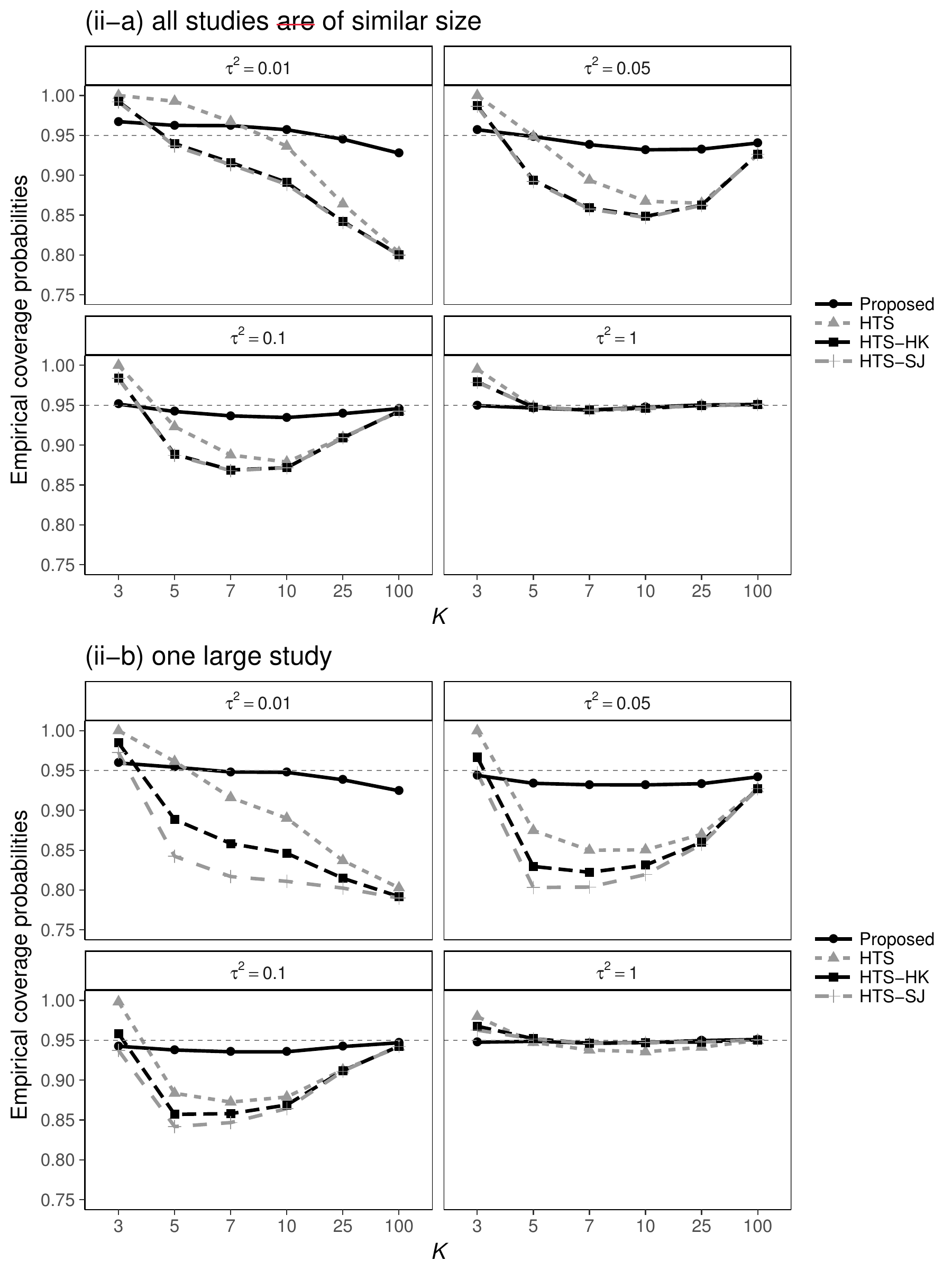}
	\caption{
		Simulation results (ii): the performance of the HTS and proposed prediction intervals (ii-a) with all studies of similar size and (ii-b) with one large study.
		The heterogeneity parameters $\tau^2=0.01,0.05,0.1$, or $1$.
		The number of simulations was 25\,000.
		Methods: Proposed, the proposed prediction interval; HTS, the HTS prediction interval; HTS-HK, the HTS-type prediction interval following REML with the Hartung--Knapp variance estimator; HTS-SJ, the HTS-type prediction interval following REML with the Sidik--Jonkman bias-corrected variance estimator.		
	}
	\label{fig:2}
\end{figure}

\begin{figure}
	\centering
	\includegraphics[width=0.9\textwidth]{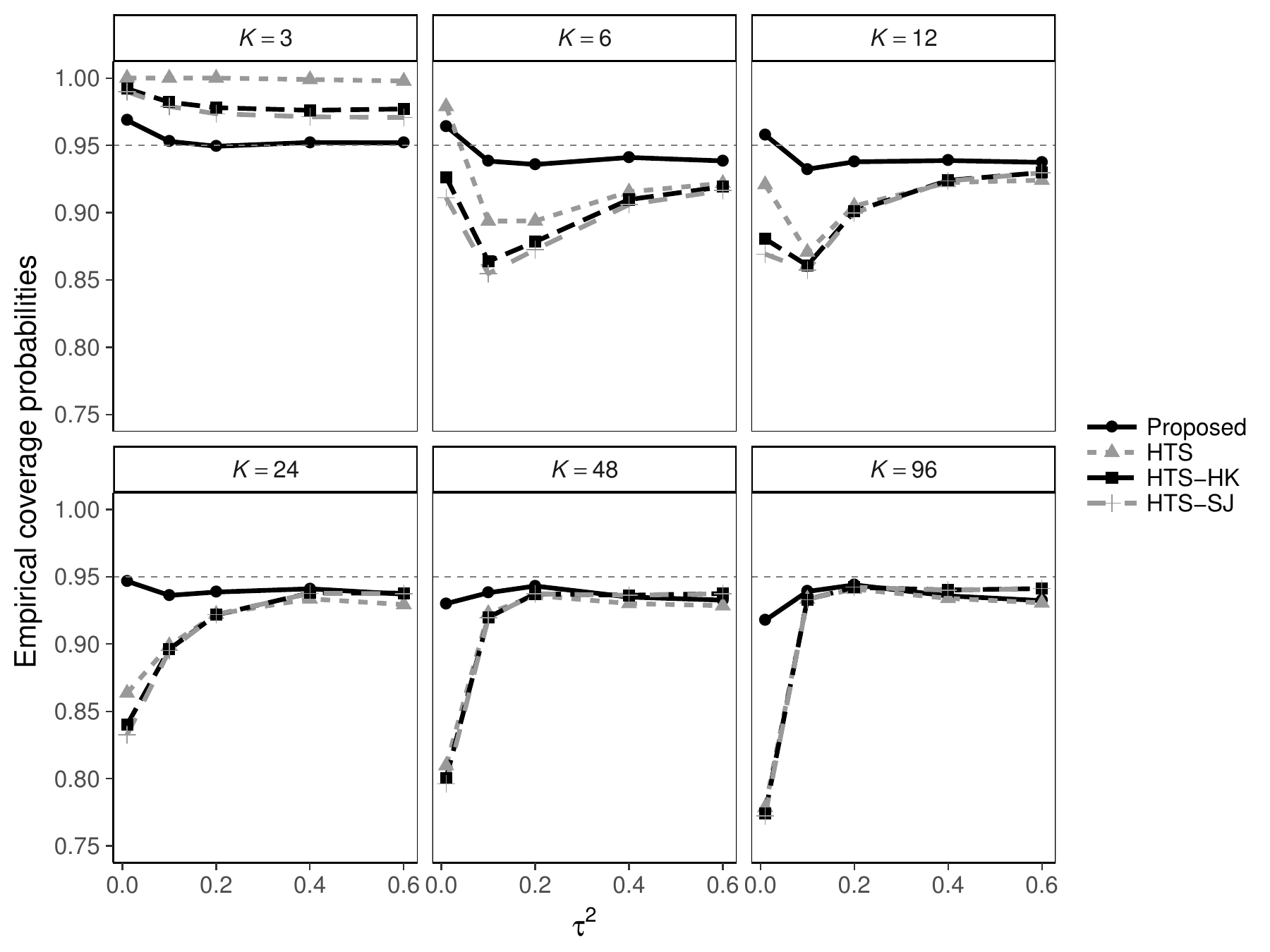}
	\caption{
		Simulation results (iii): the performance of the HTS and proposed prediction intervals for $\mu=0$.
		The number of studies $K=3,6,12,24,48$, or $96$.
		The number of simulations was 25\,000.
		Methods: Proposed, the proposed prediction interval; HTS, the HTS prediction interval; HTS-HK, the HTS-type prediction interval following REML with the Hartung--Knapp variance estimator; HTS-SJ, the HTS-type prediction interval following REML with the Sidik--Jonkman bias-corrected variance estimator.
	}
	\label{fig:3}
\end{figure}

The results of simulation (i) are presented in Figure \ref{fig:1}.
The coverage probabilities of the HTS prediction interval were approximately 90\%, far short of the nominal level of 95\%.
The under-coverage of the HTS prediction interval reflects the rough $t$-approximation; thus, the source of the problem is substitution of an estimate for $\tau^2$ or ignoring uncertainty in $\tau^2$.
The results show that the HTS-HK and HTS-SJ prediction intervals are also deficient.
The coverage probabilities for the HTS-HK and HTS-SJ prediction intervals almost retained the nominal level except in situations where the relative degree of heterogeneity is small or moderate.
For example, the coverage probabilities of the HTS-HK prediction interval were 82.8\%--98.5\% for $\tau^2=0.01$, 83.8\%--96.5\% for $\tau^2=0.05$, and 85.3\%--95.6\% for $\tau^2=0.1$; the coverage probabilities of the HTS-SJ prediction interval were 79.7\%--96.7\% for $\tau^2=0.01$, 80.9\%--93.9\% for $\tau^2=0.05$, and 83.9\%--92.9\% for $\tau^2=0.1$.
By contrast, the coverage probabilities for the proposed prediction interval almost always retained the nominal level.
The only exception was when $K = 25$ and $\tau^2=0.01$, where the coverage probability for the proposed prediction interval was 93.0\%, slightly below the nominal level.
In this case, the coverage probability for the HTS, HTS-HJ, and HTS-SJ prediction intervals were even smaller, at 82.4\%, 82.8\%, and 82.0\%, respectively.
Analyses using very few studies ($K<5$) pose problems in random-effects models, as discussed by Higgins \emph{et al.} \cite{Higgins2009}.
Nevertheless, the proposed method performed well even when $K=3$.
The nominal level was attained for nearly  all values of the heterogeneity parameter in the proposed prediction interval.

The results of simulation (ii-a), with all studies of similar size, are presented in Figure \ref{fig:2}.
The results show that all HTS prediction intervals are also deficient except for $\tau^2=0.001, 0.05$.
The coverage probabilities for all HTS prediction intervals almost retained the nominal level for $\tau^2 = 1$ and $K \geq 5$.
The coverage probabilities were too large for $K=3$ and too small for $K=5$--$25$ and $\tau^2 = 0.1$.
In the case of $\tau^2=0.01,0.05$, the coverage probabilities of the HTS-HJ and HTS-SJ prediction intervals were too small for $K=5$--$100$, and the coverage probability of the HTS prediction interval was too small for $K=10$--$100$.
By contrast, the coverage probabilities for the proposed prediction interval almost always retained the nominal level.
The only exception was when $\tau^2=0.01$, where the coverage probabilities for the proposed prediction interval were 92.8\%--96.7\%, slightly below the nominal level for $K=100$.
The results of simulation (ii-b), with one large study, are presented in Figure \ref{fig:2}.
The coverage probabilities appear to be relatively poor compared to the balanced case (ii-a) even in large heterogeneity variance not only for the HTS-HK prediction interval, but also for the HTS-SJ prediction interval.
Moreover, the performance of the HTS-SJ prediction interval was somewhat poorer (showed under coverage) compared to the HTS-HK prediction interval.
In contrast, the coverage probabilities for the proposed prediction interval nearly always retained the nominal level.
The results of simulation (ii-c), with one small study, are presented in Supplementary Figure S1.
The coverage probabilities were similar to those of the balanced case (ii-a).

The results of simulation (iii) for $\mu=0$ are presented in Figure \ref{fig:3}.
The coverage probabilities for all HTS prediction intervals were too large for $K=3$, too small for $K = 6, 12, 24, 48, 96$ and $\tau^2 \leq 0.1$, and nearly retained the nominal level for $\tau^2 = 0.6$ or $K = 96$.
In contrast, the coverage probabilities for the proposed prediction interval nearly retained the nominal level, except for $K=96$ and $\tau=0.01$.
The results of simulation (iii) for $\mu=-0.5,0.5$ are presented in Supplementary Figures S2 and S3.
The coverage probabilities for $\mu=-0.5,0.5$ were similar to those of the case for $\mu=0$.

In summary, the HTS prediction intervals have insufficient coverage, except when the relative degree of heterogeneity is large and may show severe under-coverage under realistic meta-analysis settings, possibly providing misleading results and interpretation.
In contrast, the proposed prediction interval achieves the nominal level of coverage.

\section{Applications}
\label{sec:4}
We applied the methods to the following three published random-effects meta-analyses.
\begin{enumerate}
	\renewcommand{\labelenumi}{(\theenumi)}
	\renewcommand{\theenumi}{\Alph{enumi}}
	\item Set-shifting data: Higgins \emph{et al.} \cite{Higgins2009} re-analyzed data \cite{Roberts2007} that included 14 studies evaluating the set-shifting ability in people with eating disorders by using a prediction interval.
	Standardized mean differences in the time taken to complete Trail Making Test between subjects with eating disorders and healthy controls were collected.
	Positive estimates indicate impairment in set shifting ability in people with eating disorders.
	\item Pain data: The pain data \cite{Riley2011a,Hauser2009} included 22 studies comparing the treatment effect of antidepressants on reducing pain in patients with fibromyalgia syndrome.
	The treatment effects were summarized using standardized mean differences on a visual analog scale for pain between the antidepressant group and control group.
	Negative estimates indicate the reduction of pain in the antidepressant group.
	\item Systolic blood pressure (SBP) data: Riley \emph{et al.} \cite{Riley2011a} analyzed a hypothetical meta-analysis.
	They generated a data set of 10 studies examining the same antihypertensive drug.
	Negative estimates suggested reduced blood pressure in the treatment group.
\end{enumerate}
These data sets are reproduced in Figure \ref{fig:4}.
The number of bootstrap samples $B$ was set to 50\,000.

\begin{figure}
	\centering
	\includegraphics[height=0.95\textheight]{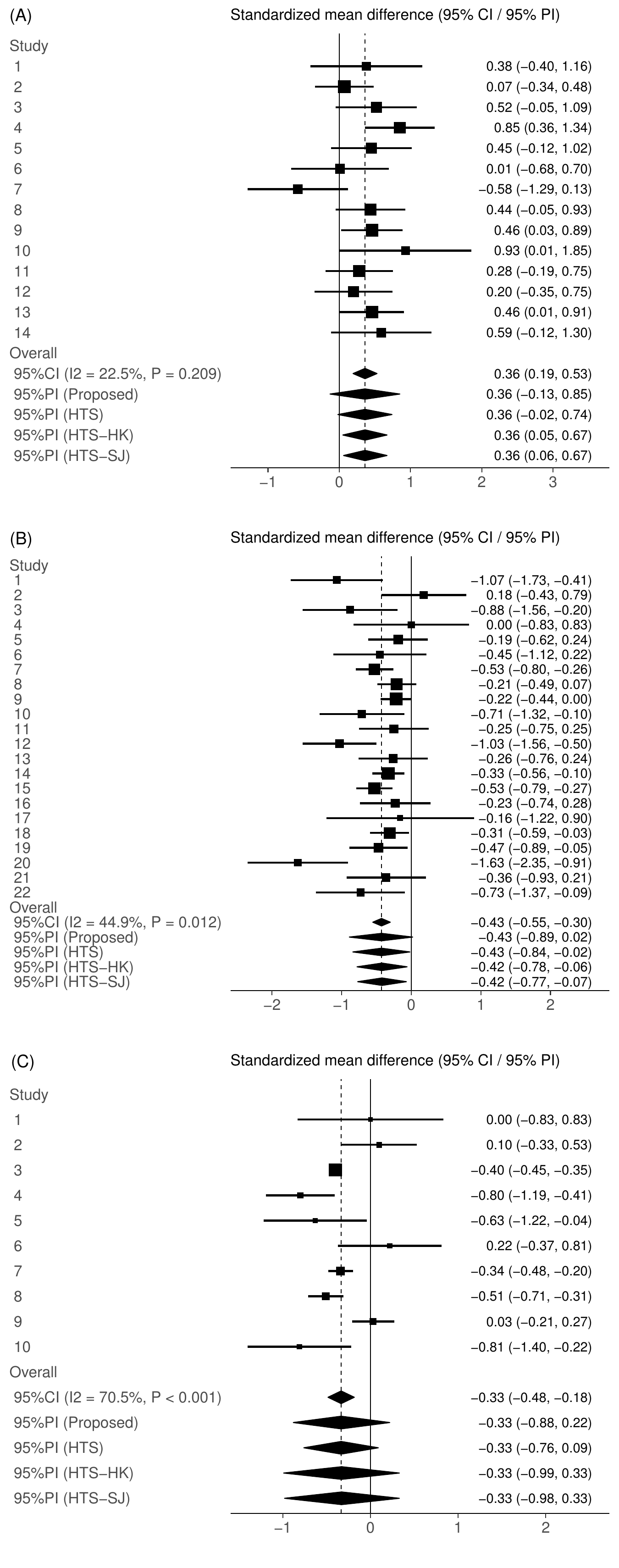}
	\caption{
		The three data sets and summary results: (A) Set-shifting data \cite{Roberts2007} ($K=14$), (B) Pain data \cite{Hauser2009} ($K=22$), and (C) SBP data \cite{Riley2011a} ($K=10$).
		Abbreviations: CI, confidence interval; PI, prediction interval.
	}
	\label{fig:4}
\end{figure}

\begin{table}
	\small\sf\centering
	\caption{
		Results from the three data sets: the average treatment effect ($\hat{\mu}$) and its 95\% confidence interval, heterogeneity measures ($\hat{\tau}_{DL}^2$, $\hat{\tau}_{R}^2$ and $\hat{I}^2$), the $P$-value for the test of heterogeneity, the proposed prediction interval, and the HTS prediction intervals.
	}
	\label{tab:1}
	\begin{tabular}{ccccc} \toprule
		\multicolumn{2}{c}{Data} & Set-shifting & Pain & SBP \\ 
		\multicolumn{2}{c}{} & ($K=14$) & ($K=22$) & ($K=10$) \\ \hline
		\multicolumn{2}{c}{$\hat{\mu}$ (DL)} & 0.36 & \llap{$-$}0.43 & \llap{$-$}0.33 \\
		\multicolumn{2}{c}{95\%CI (DL)} & [0.19, 0.53] & [$-$0.55, $-$0.30] & [$-$0.48, $-$0.18] \\ \midrule
		\multicolumn{2}{c}{$\hat{\tau}_{DL}^2$} & 0.023 & 0.034 & 0.023 \\
		\multicolumn{2}{c}{$\hat{\tau}_{R}^2$} & 0.013 & 0.025 & 0.070 \\
		\multicolumn{2}{c}{$\hat{I}^2$ (DL)} & 22.5\% & 44.9\% & 70.5\% \\
		\multicolumn{2}{c}{$P$-value for heterogeneity} & 0.209 & 0.012 & $<$0.001 \\ \midrule
		95\%PI & Proposed  & [$-$0.13, 0.85] & [$-$0.89, 0.02] & [$-$0.88, 0.23] \\
		& HTS & [$-$0.02, 0.74] & [$-$0.84, $-$0.02] & [$-$0.76, 0.09] \\
		& HTS-HK & [0.05, 0.67] & [$-$0.78, $-$0.06] & [$-$0.99, 0.33] \\
		& HTS-SJ & [0.06, 0.67] & [$-$0.77, $-$0.07] & [$-$0.98, 0.33] \\
		\midrule
		length of & Proposed  & 0.98 & 0.91 & 1.10 \\
		95\%PI & HTS & 0.76 & 0.82 & 0.85 \\
		& HTS-HK & 0.62 & 0.72 & 1.32 \\
		& HTS-SJ & 0.61 & 0.72 & 1.31 \\ \bottomrule
	\end{tabular}
\end{table}

Table \ref{tab:1} presents estimates of the average treatment effect and its confidence interval, heterogeneity measures, the $P$-value for the test of heterogeneity, the proposed prediction interval, and the HTS prediction intervals.
None of the confidence intervals for the average treatment effect included $0$ (set-shifting data: $[0.19, 0.53]$; pain data: $[-0.55, -0.30]$; SBP data: $[-0.48, -0.18]$).
This means that on average the interventions are significantly effective.
However, small, moderate, and large heterogeneity were observed in the three data sets (set-shifting data: $\hat{\tau}_{DL}^2=0.023$, $\hat{I}^2=22.5\%$; pain data: $\hat{\tau}_{DL}^2=0.185$, $\hat{I}^2=44.9\%$; SBP data: $\hat{\tau}_{DL}^2=0.023$, $\hat{I}^2=70.5\%$).
Accounting for heterogeneity, prediction intervals would provide additional relevant statistical information.
There were large differences between the 95\% confidence interval and prediction intervals, even in the case of small heterogeneity.

As shown in Figure \ref{fig:4} and summarized in Table \ref{tab:1}, the proposed prediction intervals were substantially wider than the HTS prediction intervals in all three analyses.
The proposed prediction intervals were $29\%$, $11\%$, and $31\%$ wider than the HTS prediction intervals for the set-shifting data, pain data, and SBP data, respectively.
As observed in Section \ref{sec:3}, the HTS-HK and HTS-SJ prediction intervals showed similar results.
The proposed prediction intervals were $58\%$, $27\%$ wider and $17\%$ narrower than the HTS-HK (or HTS-SJ) prediction intervals for the set-shifting data, pain data, and SBP data, respectively.
Only for the SBP data, the proposed prediction interval was narrower than the HTS-HK prediction interval; this is because the two intervals were based on different heterogeneity variance estimators and $\hat{\tau}_{DL}^2<\hat{\tau}_{R}^2$.

The prediction intervals may lead to different interpretations of the results.
In the set-shifting data, the HTS-HK and HTS-SJ prediction intervals did not include $0$, but the proposed prediction interval included $0$.
For the pain data, the HTS, HTS-HK and HTS-SJ prediction intervals did not include $0$, in a frequentist sense, suggesting that the intervention may be beneficial in most subpopulations.
In contrast, the proposed prediction interval included $0$, indicating that the intervention may not be beneficial in some subpopulations.
However, taking a Bayesian perspective, all the prediction intervals suggest that there is a large probability and the treatment will be effective in a new population.
The simulation results in Section \ref{sec:3} suggest that the HTS prediction intervals could have under-coverage in situations where the relative degree of heterogeneity is small or moderate.
Since $\hat{\tau}^2_{DL}$ of three data sets and $\hat{\tau}^2_{R}$ of set-shifting and pain data were small ($\approx 0.02$), it may be too narrow under realistic situations and may provide misleading results.

\section{Discussion and conclusion}
\label{sec:5}
For the random-effects model in meta-analysis, the average treatment effect and its confidence interval have been used with heterogeneity measures such as the $I^2$-statistic and $\tau^2$.
However, results from random-effects models have sometimes been misinterpreted.
Thus, the new concept ``prediction interval'' was proposed, which is useful in applying the results to other subpopulations and in decision making.
The HTS prediction intervals have a theoretical problem, namely that its rough $t$-approximation could have a detrimental impact on the coverage probability.
We have presented an appropriate prediction interval to account for the uncertainty in $\tau^2$ by using a confidence distribution.

Simulation studies showed that the HTS prediction intervals could have severe under-coverage for realistic meta-analysis settings and might lead to misleading results and interpretation.
The simulation results suggested that the HTS prediction interval may be too narrow when analyzing a small number of studies.
This interval is valid when $K \ge 25$, but in many meta-analyses $K$ is much smaller than 25.
The HTS-HK and HTS-SJ prediction intervals may be too narrow when the relative degree of heterogeneity is small.
By contrast, the coverage probabilities for the proposed prediction interval satisfactorily retained the nominal level.
Although Higgins \emph{et al.} \cite{Higgins2009} cautioned that the random-effects model may not work well under very small numbers of studies ($K<5$), the proposed method performed well even when $K=3$.
Since the heterogeneity parameter had very little effect on the performance of the proposed prediction interval, the method would be valid regardless of the value of the heterogeneity parameter.
Moreover, all prediction intervals (i.e., the random-effects model in \eqref{equ:1}) assume normality of the between-study distribution of true effects, $u_k \sim N(0, \tau^2)$, but the assumption may not be true in practice.
A full Bayesian approach may be useful for constructing a suitable prediction interval.

Applications to the three published random-effects meta-analyses concluded that substantially different results and interpretation might be obtained from the prediction intervals.
Since the HTS prediction interval is always narrower and the HTS-HK and HTS-SJ prediction intervals are narrower when the heterogeneity parameter is small or moderate, we should be cautious in using and interpreting these approaches.

In conclusion, we showed that the proposed prediction interval works well and is suitable for random-effects meta-analysis.
As shown in the three illustrative examples, quite different results and interpretations are obtained using our new method.
Extensions of these results to other complicated models such as network meta-analysis are now warranted.

\section*{Acknowledgements}
This study was supported by CREST from the Japan Science and Technology Agency (Grant number: JPMJCR1412).

\newcommand{\noopsort}[1]{}

\clearpage
\appendix
\appendix
\section*{Proof of Theorem 1}

\begin{proof}[Proof of Theorem \ref{the:1}]
	(R1) Since $F_T$ is a continuous distribution function, $H(\phi)=1-F_T(T(\boldsymbol{\mathrm{y}}); \phi)$ is continuous on $\boldsymbol{\mathrm{Y}}\times \Phi \to [0, 1]$.
	(R2) By the continuity of $F_T$, a derivative, $g(\phi)=\mathrm{d} F_T(T(\boldsymbol{\mathrm{y}}); \phi)/\mathrm{d} \phi$, exists, and $G(\phi)=\int g(\phi) \,\mathrm{d} \phi=F_T(T(\boldsymbol{\mathrm{y}}); \phi)$.
	By (R1) and the monotone decreasingness of $F_T$, $G(\phi_{\min})=1$ and $G(\phi_{\max})=0$.
	Therefore, $H(\phi)$ can be written as $1-\int_{\phi}^{\phi_{\max}} -g(s) \,\mathrm{d} s=1-G(\phi)$.
	Writing $h(\phi)=-g(\phi)$, we find $1-\int_{\phi}^{\phi_{\max}} h(s) \,\mathrm{d} s=\int_{\phi_{\min}}^{\phi} h(s) \,\mathrm{d} s$.
	Thus, $H(\phi)$ is clearly a cumulative distribution function on $\phi$.
	(R3) At the true parameter value $\phi=\phi_0$, it follows that $1-F_T(T(\boldsymbol{\mathrm{y}}); \phi_0) \sim U(0, 1)$.
	Thus, by Definition \ref{def:1}, $H(\phi)$ is a confidence distribution for the parameter $\phi$, and $h(\phi)$ is a confidence density function for $\phi$.
\end{proof}

\section*{Supplementary Figures}
\setcounter{figure}{0}
\def\thefigure{S\arabic{figure}}

\begin{enumerate}
	\setlength{\itemindent}{36pt}
	\setlength{\labelwidth}{66pt}
	\item[\textbf{Figure S1}.] Simulation results (ii): the performance of the HTS and proposed prediction intervals with one small study.
	\item[\textbf{Figure S2}.] Simulation results (iii): the performance of the HTS and proposed prediction intervals for $\mu=0.5$.
	\item[\textbf{Figure S3}.] Simulation results (iii): the performance of the HTS and proposed prediction intervals for $\mu=-0.5$.
\end{enumerate}

\begin{figure}[h]
	\centering
	\includegraphics[width=0.9\textwidth]{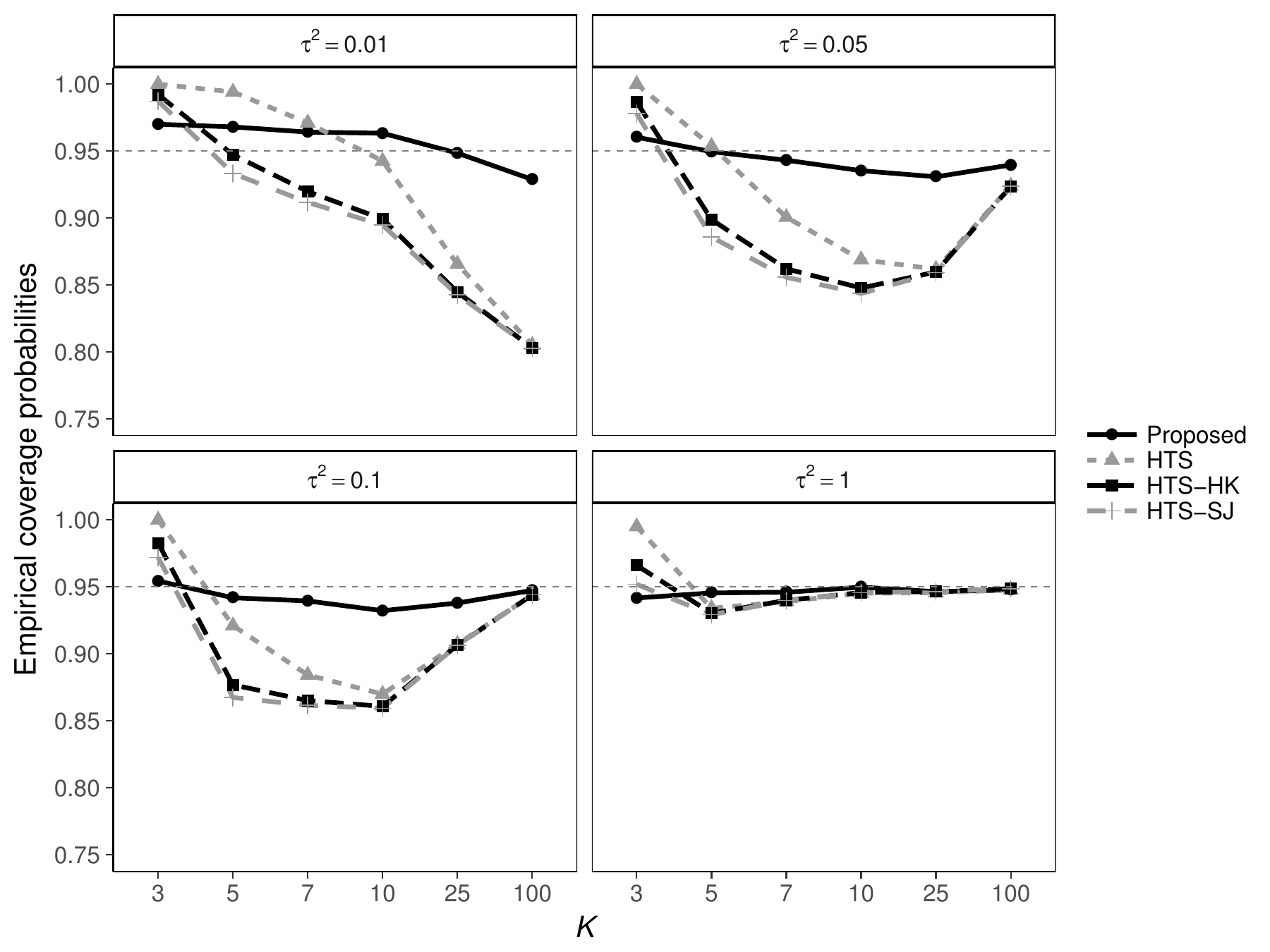}
	\caption{
		Simulation results (ii): the performance of the HTS and proposed prediction intervals with one small study.
		The heterogeneity parameters $\tau^2=0.01,0.05,0.1$, or $1$.
		The number of simulations was 25\,000.
		Methods: Proposed, the proposed prediction interval; HTS, the HTS prediction interval; HTS-HK, the HTS-type prediction interval following REML with the Hartung--Knapp variance estimator; HTS-SJ, the HTS-type prediction interval following REML with the Sidik--Jonkman bias-corrected variance estimator.		
	}
	\label{fig:s1}
\end{figure}

\begin{figure}[h]
	\centering
	\includegraphics[width=0.9\textwidth]{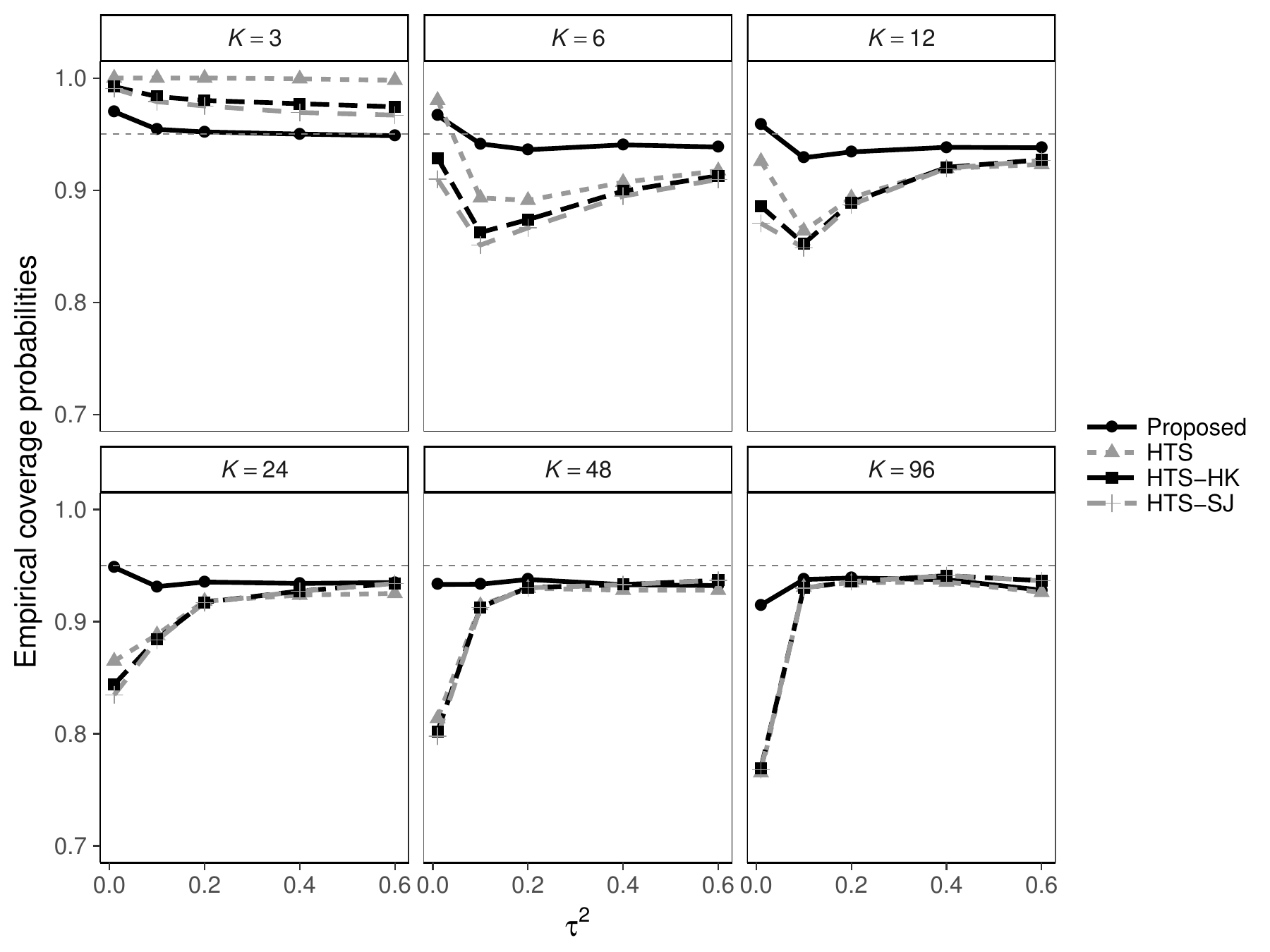}
	\caption{
		Simulation results (iii): the performance of the HTS and proposed prediction intervals for $\mu=0.5$.
		The number of studies $K=3,6,12,24,48$, or $96$.
		The number of simulations was 25\,000.
		Methods: Proposed, the proposed prediction interval; HTS, the HTS prediction interval; HTS-HK, the HTS-type prediction interval following REML with the Hartung--Knapp variance estimator; HTS-SJ, the HTS-type prediction interval following REML with the Sidik--Jonkman bias-corrected variance estimator.	
	}
	\label{fig:s2}
\end{figure}

\begin{figure}[h]
	\centering
	\includegraphics[width=0.9\textwidth]{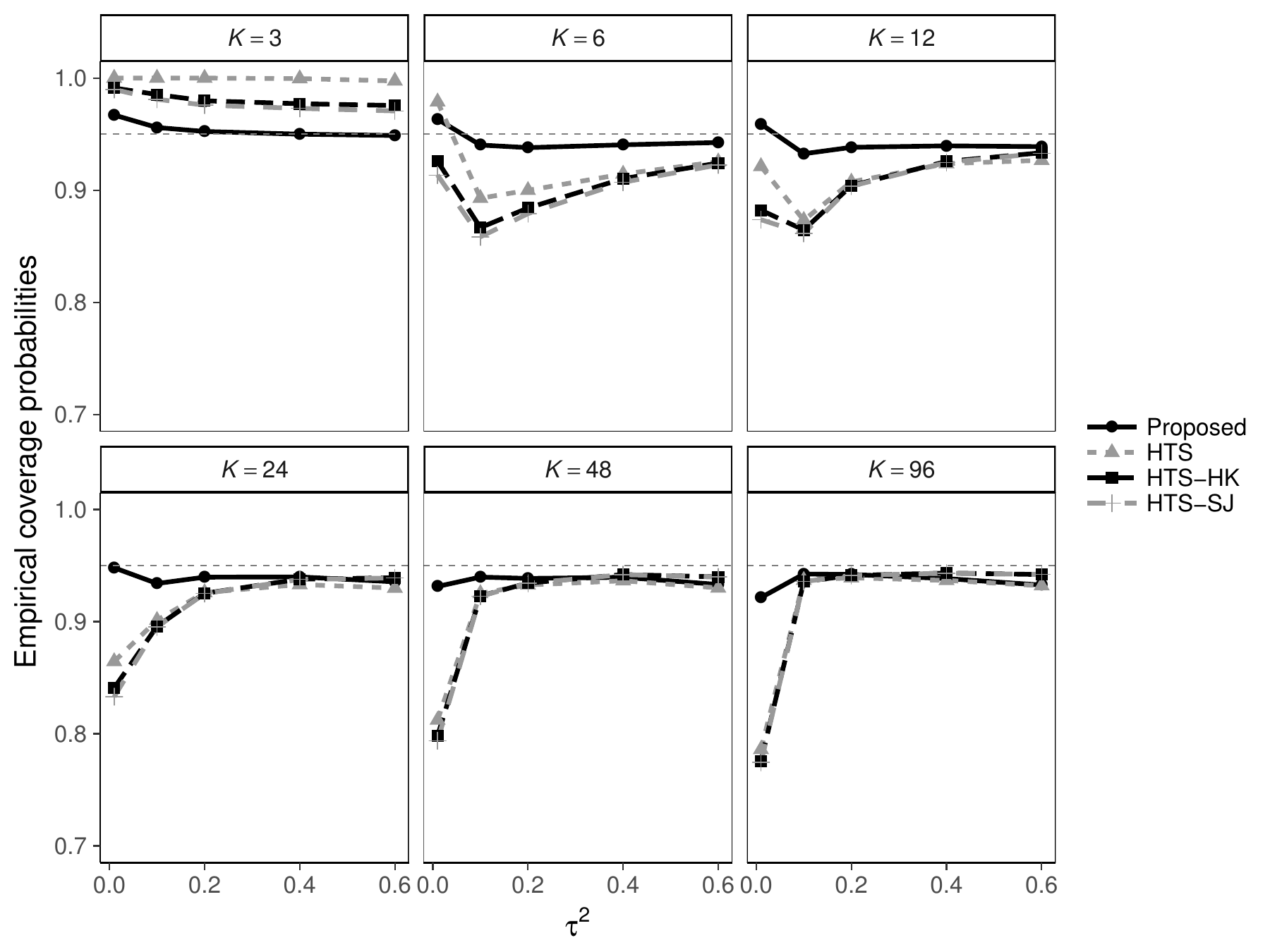}
	\caption{
		Simulation results (iii): the performance of the HTS and proposed prediction intervals for $\mu=-0.5$.
		The number of studies $K=3,6,12,24,48$, or $96$.
		The number of simulations was 25\,000.
		Methods: Proposed, the proposed prediction interval; HTS, the HTS prediction interval; HTS-HK, the HTS-type prediction interval following REML with the Hartung--Knapp variance estimator; HTS-SJ, the HTS-type prediction interval following REML with the Sidik--Jonkman bias-corrected variance estimator.
	}
	\label{fig:s3}
\end{figure}


\begin{thebibliography}{10}
	\providecommand{\url}[1]{\texttt{#1}}
	\providecommand{\urlprefix}{URL }
	\expandafter\ifx\csname urlstyle\endcsname\relax
	\providecommand{\doi}[1]{DOI:\discretionary{}{}{}#1}\else
	\providecommand{\doi}{DOI:\discretionary{}{}{}\begingroup
		\urlstyle{rm}\Url}\fi
	
	\bibitem{Borenstein2009}
	Borenstein M, Hedges LV, Higgins JPT, et al.
	\textit{Introduction to Meta-Analysis}.
	Chichester: Wiley, 2009.
	
	\bibitem{Higgins2002}
	Higgins JPT and Thompson SG.
	Quantifying heterogeneity in a meta-analysis.
	\textit{Stat Med} 2002; \textbf{21}: 1539--1558.
	
	\bibitem{Higgins2003}
	Higgins JPT, Thompson SG, Deeks JJ, et al.
	Measuring inconsistency in meta-analyses.
	\textit{BMJ} 2003; \textbf{327}: 557--560.
	
	\bibitem{Riley2011a}
	Riley RD, Higgins JPT and Deeks JJ.
	Interpretation of random effects meta-analyses.
	\textit{BMJ} 2011; \textbf{342}: d549.
	
	\bibitem{Riley2011b}
	Riley RD, Gates SG, Neilson J, et al.
	Statistical methods can be improved within Cochrane pregnancy and childbirth reviews.
	\textit{J Clin Epidemiol} 2011; \textbf{64}: 608--618.
	
	\bibitem{Higgins2009}
	Higgins JPT, Thompson SG and Spiegelhalter DJ.
	A re-evaluation of random-effects meta-analysis.
	\textit{J R Stat Soc Ser A Stat Soc} 2009; \textbf{172}: 137--159.
	
	\bibitem{Brockwell2001}
	Brockwell SE and Gordon IR.
	A comparison of statistical methods for meta-analysis.
	\textit{Stat Med} 2001; \textbf{20}: 825--840.
	
	\bibitem{Noma2011}
	Noma H.
	Confidence intervals for a random-effects meta-analysis based on Bartlett-type corrections.
	\textit{Stat Med} 2011; \textbf{30}: 3304--3312.
	
	\bibitem{Kontopantelis2013}
	Kontopantelis E, Springate DA and Reeves D.
	A re-analysis of the Cochrane library data: the dangers of unobserved heterogeneity in meta-analyses.
	\textit{PLoS One} 2013; \textbf{8}: e69930.
	
	\bibitem{Partlett2017}
	Partlett C and Riley RD.
	Random effects meta-analysis: Coverage performance of 95\% confidence and prediction intervals following REML estimation.
	\textit{Stat Med} 2017; \textbf{36}: 301--317.
	
	\bibitem{Schweder2002}
	Schweder T and Hjort NL.
	Confidence and likelihood.
	\textit{Scand J Stat} 2002; \textbf{29}: 309--332.
	
	\bibitem{Schweder2016}
	Schweder T and Hjort NL.
	\textit{Confidence, Likelihood, Probability: Statistical Inference with Confidence Distributions}.
	New York: Cambridge University Press, 2016.
	
	\bibitem{Singh2005}
	Singh K, Xie M and Strawderman WE.
	Combining information from independent sources through confidence distributions.
	\textit{Ann Stat} 2005; \textbf{33}: 159--183.
	
	\bibitem{Singh2007}
	Singh K, Xie M and Strawderman WE.
	Confidence distribution (CD)-distribution estimator of a parameter.
	In \textit{Complex Datasets and Inverse Problems: Tomography, Networks, and Beyond}.
	Liu R, Strawderman WE and Zhang CH (eds).
	Beachwood: Institute of Mathematical Statistics, 2007; \textbf{54}: 132--150.
	
	\bibitem{Xie2013}
	Xie M and Singh K.
	Confidence distribution, the frequentist distribution estimator of a parameter: a review.
	\textit{Int Stat Rev} 2013; \textbf{81}: 3--39.
	
	\bibitem{Cochran1937}
	Cochran WG.
	Problems arising in the analysis of a series of similar experiments.
	\textit{J R Stat Soc} 1937; (Supplment) \textbf{4}:102--118.
	
	\bibitem{Cochran1954}
	Cochran WG.
	The combination of estimates from different experiments.
	\textit{Biometrics} 1954; \textbf{10}: 101--129.
	
	\bibitem{DerSimonian1986}
	DerSimonian R and Laird N.
	Meta-analysis in clinical trials.
	\textit{Control Clin Trials} 1986; \textbf{7}: 177--188.
	
	\bibitem{Whitehead1991}
	Whitehead A and Whitehead J.
	A general parametric approach to the meta-analysis of randomized clinical trials.
	\textit{Stat Med} 1991; \textbf{10}: 1665--1677.
	
	\bibitem{Biggerstaff1997}
	Biggerstaff BJ and Tweedie RL.
	Incorporating variability of estimates of heterogeneity in the random effects model in meta-analysis.
	\textit{Stat Med} 1997; \textbf{16}: 753--768.
	
	\bibitem{Biggerstaff2008}
	Biggerstaff BJ and Jackson D.
	The exact distribution of Cochran's heterogeneity statistic in one-way random effects meta-analysis.
	\textit{Stat Med} 2008; \textbf{27}: 6093--6110.
	
	\bibitem{Sidik2007}
	Sidik K and Jonkman JN.
	A comparison of heterogeneity variance estimators in combining results of studies.
	\textit{Stat Med} 2007; \textbf{26}: 1964--1981.
	
	\bibitem{Scheffe1959}
	Scheff\'e H.
	\textit{The Analysis of Variance}.
	New York: Wiley, 1959.
	
	\bibitem{Graybill1976}
	Graybill FA.
	\textit{Theory and Application of the Linear Model}.
	North Scituate: Duxbury Press, 1976.
	
	\bibitem{Mathai1992}
	Mathai AM and Provost SB.
	\textit{Quadratic Forms in Random Variables: Theory and Applications}.
	New York: Marcel Dekker, 1992.
	
	\bibitem{Hartung2001}
	Hartung J and Knapp G.
	On tests of the overall treatment effect in meta-analysis with normally distributed responses.
	\textit{Stat Med} 2001; \textbf{20}: 1771--1782.
	
	\bibitem{Sidik2006}
	Sidik K and Jonkman JN.
	Robust variance estimation for random effects meta-analysis.
	\textit{Comput Stat Data Anal} 2006; \textbf{50}: 3681--3701.
	
	\bibitem{Harville1977}
	Harville DA.
	Maximum likelihood approaches to variance component estimation and to related problems.
	\textit{J Am Stat Assoc} 1977; \textbf{72}: 320--339.
	
	\bibitem{Raudenbush1985}
	Raudenbush SW and Bryk AS.
	Empirical Bayes meta-analysis.
	\textit{J Educ Stat} 1985; \textbf{10}: 75--98.
	
	\bibitem{Hartung1999}
	Hartung J.
	An alternative method for meta-analysis.
	\textit{Biom J} 1999; \textbf{41}: 901--916.
	
	\bibitem{Fisher1935}
	Fisher RA.
	The fiducial argument in statistical inference.
	\textit{Ann Eugen} 1935; \textbf{6}: 391--398.
	
	\bibitem{Efron1998}
	Efron B.
	R. A. Fisher in the 21st century.
	\textit{Stat Sci} 1998; \textbf{13}: 95--122.
	
	\bibitem{Viechtbauer2007}
	Viechtbauer W.
	Confidence intervals for the amount of heterogeneity in meta-analysis.
	\textit{Stat Med} 2007; \textbf{26}: 37--52.
	
	\bibitem{Forsythe1977}
	Forsythe GE, Malcolm MA and Moler CB.
	\textit{Computer Methods for Mathematical Computations}.
	Englewood Cliffs: Prentice-Hall, 1977.
	
	\bibitem{Farebrother1984}
	Farebrother RW.
	Algorithm AS 204: the distribution of a positive linear combination of $\chi^2$ random variables.
	\textit{J R Stat Soc Ser C Appl Stat} 1984; \textbf{33}: 332--339.
	
	\bibitem{Brockwell2007}
	Brockwell SE and Gordon IR.
	A simple method for inference on an overall effect in meta-analysis.
	\textit{Stat Med} 2007; \textbf{26}: 4531--4543.
	
	\bibitem{Jackson2013}
	Jackson D.
	Confidence intervals for the between study variance in random effects meta-analysis using generalized Cochran heterogeneity statistics.
	\textit{Res Synth Methods} 2013; \textbf{4}: 220--229.
	
	\bibitem{Turner2012}
	Turner RM, Davey J, Clarke MJ, et al.
	Predicting the extent of heterogeneity in meta-analysis, using empirical data from the Cochrane Database of Systematic Reviews.
	\textit{Int J Epidemiol} 2012; \textbf{41}: 818--827.
	
	\bibitem{Rhodes2015}
	Rhodes KM, Turner RM and Higgins JPT.
	Predictive distributions were developed for the extent of heterogeneity in meta-analyses of continuous outcome data.
	\textit{J Clin Epidemiol} 2015; \textbf{68}: 52--60.
	
	\bibitem{Roberts2007}
	Roberts ME, Tchanturia K, Stahl D, et al.
	A systematic review and meta-analysis of set-shifting ability in eating disorders.
	\textit{Psychol Med} 2007; \textbf{37}: 1075--1084.
	
	\bibitem{Hauser2009}
	H{\"a}user W, Bernardy K, {\"U}\c{c}eyler N, et al.
	Treatment of fibromyalgia syndrome with antidepressants: a meta-analysis.
	\textit{JAMA} 2009; \textbf{301}: 198--209.
	
\end{thebibliography}
\end{document}